\documentclass[envcountsame]{llncs} 
\usepackage[english]{babel}
\usepackage[colorlinks=true,linkcolor=black,citecolor=black,filecolor=black]{hyperref}
\usepackage{color}
\usepackage{xcolor}
\usepackage{amsmath,amssymb,enumitem,xspace,tikz,stmaryrd,array}
\usetikzlibrary{arrows,automata,shapes,positioning}

\usepackage{graphicx,caption,subcaption,pslatex}
\newcommand{\comment}[1]{\marginpar{\phantom{\small{#1}}}}

\newcommand{\arena}{\ensuremath{\mathcal{A}}}
\newcommand{\Win}{\ensuremath{\mathit{Win}}}
\newcommand{\Reach}{\ensuremath{\mathit{Reach}}}
\newcommand{\Outcome}{\ensuremath{\mathit{Outcome}}}
\newcommand{\N}{\ensuremath{\mathbb{N}}}
\newcommand{\view}{\ensuremath{\mathit{view}}}
\newcommand{\Views}{\ensuremath{\mathcal{V}}}
\newcommand{\obs}{\ensuremath{\mathit{obs}}}
\newcommand{\equivclass}[1]{[#1]_\equiv}
\newcommand{\Act}{\ensuremath{\mathbb{M}}}

\newcommand{\Config}{\ensuremath{\mathcal{C}}}
\newcommand{\Agather}{\ensuremath{\arena_{\textrm{gather}}}}
\newcommand{\rep}{\ensuremath{\mathit{rep}}}
\newcommand{\Sum}[3]{\sum\limits_{#1}^{#2} #3}

\newcommand{\clockwise}{\ensuremath{\curvearrowright}}
\newcommand{\counterclockwise}{\ensuremath{\curvearrowleft}}
\newcommand{\still}{\ensuremath{\uparrow}}

\newcommand{\disoriented}{?}
\newcommand{\Actions}{\ensuremath{\Delta}}

\newcommand{\PosTower}{\ensuremath{\mathit{PosTower}}}
\newcommand{\Pos}{\ensuremath{\mathit{Pos}}}
\usepackage[]{algorithm2e}
\usepackage{wrapfig}

\newenvironment{theorem-repeat}[1]{\begin{trivlist}
\item[\hspace{\labelsep}{\bf\noindent Theorem~\ref{#1} }]}%
{\end{trivlist}}
\newcommand{\toto}{xxx}

\title{On the Synthesis of Mobile Robots Algorithms: \\
the Case of Ring Gathering}
\author{Laure Millet\inst{1,2} \and Maria Potop-Butucaru\inst{1,2} \and Nathalie Sznajder\inst{1,2} \and Sebastien Tixeuil\inst{1,2,3}}
\institute{Sorbonne Universit\'es, UPMC Univ Paris 06, UMR 7606, LIP6, F-75005, Paris, France
\and CNRS, UMR 7606, LIP6, F-75005, Paris, France
\and Institut Universitaire de France}

\begin{document}
\pagestyle{plain}

\thispagestyle{empty}
  
\maketitle

\begin{abstract}
Recent advances in Distributed Computing highlight models and algorithms for autonomous swarms of mobile robots that self-organize and cooperate to solve global objectives. The overwhelming majority of works so far considers \comment{consists in?}handmade algorithms and correctness proofs. 

 This paper is the first to propose a formal framework to automatically design distributed algorithms that are dedicated to autonomous mobile robots evolving in a discrete space. As a case study, we consider the problem of gathering all robots at a particular location, not known beforehand. Our contribution is threefold. First, we propose an encoding of the gathering problem as a reachability game. Then, we automatically generate an optimal distributed algorithm for three robots evolving on a fixed size uniform ring. Finally, we prove by induction that the generated algorithm is also correct for any ring size except when an impossibility result holds (that is, when the number of robots divides the ring size).
\end{abstract}


\section{Introduction}
\label{sec:in}
The Distributed Computing community, motivated by the variety of tasks that can be performed by autonomous robots and their complexity, started recently to propose formal models for these systems and to design and prove protocols in these models. The seminal paper by Suzuki \& Yamashita~\cite{suzuki_distributed_1996} proposes a robot model, two execution models, and several algorithms (with associated correctness proofs) for gathering and scattering a set of robots. In their model, robots are identical and anonymous (they execute the same deterministic algorithm and they cannot be distinguished using their appearance), robots are oblivious (they have no memory of their past actions) and they have neither a common sense of direction, nor a common handedness (chirality). Furthermore robots do not communicate in an explicit way. However they have the ability to sense the environment and see the position of the other robots, which lets  them find their way in their environment. Also, robots execute three-phase cycles: \textit{Look}, \textit{Compute} and \textit{Move}. During the \textit{Look} phase robots take a snapshot of the other robots' positions. The collected information is used in the \textit{Compute} phase in which robots decide to move or to stay idle. In the \textit{Move} phase, robots may move to a new position computed in the previous phase. The two execution models are denoted (using recent taxonomy~\cite{FPS12b}) FSYNC, for fully synchronous, and SSYNC, for semi-synchronous. In the SSYNC model an arbitrary non-empty subset of robots execute the three phases synchronously and atomically. In the FSYNC model all robots execute the three phases synchronously. 
 
A recent trend, motivated by practical applications such that exploration or surveillance, is the study of robots evolving  in a discrete space with a \textit{finite} number of locations. This discrete space is modeled by a graph, where nodes represent locations or sites, and edges represent the possibility for a robot to move from one site to the other. The discrete setting significantly increases the number of symmetric configurations when the underlying graph is also symmetric (\emph{e.g.} a ring). 

One of the benchmarking~\cite{FPS12b} problems for mobile robots evolving in a discrete space is that of \emph{gathering}. Regardless of their initial positions, robots have to move in such a way that they are eventually located on the same location, not known beforehand, and remain there thereafter. The case of ring networks is especially intricate, since its regular structure introduces a number of possible symmetric situations, from which the limited abilities of robots make it difficult to escape. A particular disposal (or configuration) of robots in the ring is \emph{symmetrical} if there exists an axis of symmetry, that maps single robots into single robots, multiplicities into multiplicities, and empty nodes into empty nodes. A symmetric configuration can be edge-edge, node-edge or node-node symmetrical if the axis goes through two edges, through one node and one edge, or through two nodes, respectively. A \emph{periodic} configuration is a configuration that is invariant by non-trivial rotation. 

On the negative side, it was shown~\cite{markou} that gathering is impossible when the algorithm run by every robot is deterministic and there are only two robots, or if the initial configurations are periodic, or edge-edge symmetric, or if the ability for a robot to detect multiple robots on a single location (denoted as \emph{multiplicity detection}) is not available. Running a probabilistic algorithm~\cite{OT12c} permits to start from an arbitrary initial configuration (including periodic and edge-edge symmetric) but still requires multiplicity detection. In the deterministic setting, a number of ring gathering algorithms have been proposed in the literature~\cite{KameiLOT12,navarradisc2012,navarraipdps2013,navarrasirocco2013,navarra2014} for the cases left open by impossibility results, focusing on the problem solvability for different initial configurations and different values for the size of the ring and the number of robots. When the robots are able to fully detect the number of robots in each location, a unified strategy was proposed~\cite{navarradisc2012}. When multiplicity detection is only available on the current position of each robot, more involved and specific approaches~\cite{IIKO10c,KameiLOT11,KameiLOT12,navarra2014} are needed. Every aforementioned deterministic solution considers problem solvability with particular hypotheses, and does not consider performance issues (such as time needed to reach gathering, or the total number of moves before gathering is achieved). Also, only a handmade approach for both algorithm design and proof of correctness was considered in those works.

Most related to our concern are recent approaches to mechanizing the algorithm design or the correctness proof in the context of autonomous mobile robots~\cite{BDPPT12c,DLPRT12c,BMPTT13c,ABCTU13c}. Model-checking proved useful to find bugs in existing litterature~\cite{BMPTT13c} and formally assess published algorithms~\cite{DLPRT12c,BMPTT13c}. Proof assistants enabled the use of high order logic to certify impossibility results~\cite{ABCTU13c}. To our knowledge, the only previous attempt to automatically generate mobile robots algorithms (for the problem of perpetual exclusive exploration) is due to Bonnet \emph{et al.}~\cite{BDPPT12c}, but exhibits important limitations for studying the gathering problem. Indeed their approach is brute force (it generate every possible algorithm in a particular setting, regardless of the problem to solve) and specific to configurations where \emph{(i)} a location can only host one robot (so, gathering cannot be expressed), and \emph{(ii)} no symmetry appears.
 
\textit{Games and protocols synthesis.} In the formal methods community, 
automatically synthesizing programs that would be correct by design is a problem that raised interest early~\cite{emersonclarke:1981,MannaWolper84,AbadiLamportWolper89,PnueliRosner89a}. Actually, this problem goes back to Church~\cite{Church62,BuchiLandweber69}. When the program to generate is intended to work in an open system, maintaining an on-going interaction with a (partially) unknown environment, it is known since \cite{BuchiLandweber69} that seeing the problem as a \emph{game} between the system and the environment is a successful approach. The system and its environment are considered as opposite players that play a game on some graph, the winning condition being the specification the system should fulfill however the environment behave. Then, the classical problem in game theory of determining winning strategies for the players is equivalent to find how the system should act in any situation, in order to always satisfy its specification. The case of mobile autonomous robots that we focus on in this paper falls in this category of problems: the robots may evolve (possibly indefinitely) on the ring, making decisions based on the global state of the system at each time instant. The vertices of graph on which the players will play would then be some representation of the different global positions of the robots on the ring. The presence of an opposite player (or environment) is motivated by the absence of chirality of the robots: when a robot is on an axis of symmetry, it is unable to distinguish its two sides one from another, hence to choose exactly \emph{where}
it moves ; this decision is supposed to be taken by the opposite player.

\textit{Our contribution.} In this paper, we introduce the use of formal methods for automatic synthesis of autonomous mobile robot algorithms, in the discrete space model. As a case study, we consider the problem of gathering all robots at a particular location, not known beforehand. Our contribution is threefold. First, we propose an encoding of the gathering problem as a reachability game, the players being the robot algorithm on the one side and the scheduling adversary (that is also capable for dynamically deciding robot chirality at every activation) on the other side. Our encoding is general enough to encompass classical FSYNC and SSYNC execution models for robots evolving on ring-shaped networks, including (and contrary to the existing ad hoc solution~\cite{BDPPT12c}) when several robots are located at the same node and when symmetric situations occurs. Then, in the FSYNC model, we automatically generate an \emph{optimal} distributed algorithm for three robots evolving on a fixed size uniform ring. Our optimality criterion refers to the number of robot moves that are necessary to actually achieve gathering. Finally, we prove by induction that the mechanically generated algorithm is also correct for any ring size except when an impossibility result holds (that is, when the number of robots divides the ring size). Our method can be seen as a first step towards  ``correct by design'' actual robot protocol implementations. 
 
\section{Background}
In this section we present a formal model for a robot system evolving on a ring and definitions and notations for a reachability game.

\subsection{Robot Network model}
In the following we present the robots and system model using the formalism we proposed in \cite{BMPTT13c}. 
We consider a set of robots evolving on a ring.

\paragraph{Robot model}
 A robot behavior can be described by a finite automaton.
%
%
Each robot executes a three-phase cycle composed of \textit{Look}, \textit{Compute}, \textit{Move} phases. 
To start a cycle, a robot  takes a snapshot of its environment, which is
represented by a \textit{Look} transition. Then it computes its future
movement (\textit{Compute} transition). Finally the robot moves according to its previous computation, this effective movement is
represented by a \textit{Move} transition, going back to its initial state. On a ring there are only
three possibilities for the move: stay idle, move in the clockwise direction
or in the counterclockwise direction. Note also that Look and Compute states can be merged in a single state - LookCompute.

\paragraph{Scheduler model.}
 The three existing asynchrony models fully synchronous (FSYNC), semi-synchronous (SSYNC) and asynchronous (ASYNC) in robot networks are called schedulers. The scheduler can be modeled by a finite automaton. The synchronization of these
schedulers with robots automata is an automaton that represents
the global behavior of robots in the chosen model.   

In the sequel we denote by \textit{LookCompute$_i$} (respectively \textit{Move$_i$}), the
\textit{LookCompute} (resp. \textit{Move}) phase of i$^{th}$ robot. And for a subset
$\mathit{Sched}$ of robots, we denote by $\prod\limits_{i \in
\mathit{Sched}}LookCompute_i$ (resp. $\prod\limits_{i \in
\mathit{Sched}}\mathit{Move}_i$) the synchronization of all \textit{LookCompute$_i$}
(resp. \textit{Move$_i$}) actions of all robots in $\mathit{Sched}$.

In the SSYNC model, an arbitrary non-empty subset of robots is
scheduled for execution at every phase, and operations are executed
synchronously. In this case, the automaton consists of a cycle, where
a set "\textit{Sched}" is first chosen, then the \textit{LookCompute} and
\textit{Move} phases are synchronized for this set. A generic
automaton for SSYNC is described in Figure~\ref{fig:schedSYm}.
 
The FSYNC model is a particular case of the SSYNC model, where all
robots are scheduled for execution at every phase, and operate
synchronously thereafter. 

\begin{figure}[t] 
\begin{center} 
	\begin{tikzpicture}[thick, node distance= 10em,>=stealth,->, scale=0.75, transform shape]
	\node[ellipse, minimum height=3em,minimum width= 5em, draw](a){};
	\node[text width=4em, text centered] (a1){\textit{Move} Done}; 
	\node[ellipse, minimum height=3em,minimum width= 5em, draw, right=6 em  of a](b){};
	\node[text width=4em, text centered] (b1) at(b){\textit{Sched} chosen}; 
	\node[ellipse, minimum height=3em,minimum width= 8em, draw, right= 9em of b](c){};
	\node[text width=4em, text centered, right =10 em of b] (c1){\textit{LookCompute} Done}; 
	\node (init) [below left of=a, node distance= 5em] {};

	\path 	(a) edge[] node[above=1em]{$\mathit{Choose~Sched}$} (b) 
			(b) edge[] node[above]{${\prod\limits_{i \in \mathit{Sched}}} LookCompute_i $} (c) 
			(c) edge[bend left] node[below]{${\prod\limits_{i \in \mathit{Sched}}} \mathit{Move}_i $} (a)
			(init) edge[] node{} (a); 
	\end{tikzpicture} 
\caption{The Semi-Synchronous Schedulers automaton} 
\label{fig:schedSYm} 
\end{center} \end{figure}
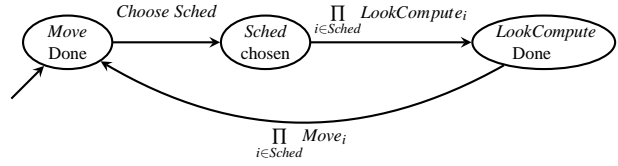 
 
\paragraph {System model.}
A configuration of $k$ robots on a ring of size $n$ encodes the position of the robots in the ring.
The system is modeled by the automaton obtained
by the synchronized product of $k$ robot automata and the possible
configurations.   The scheduler
is used to define the synchronization function.   The alphabet of
actions is $A = {\prod_{i}}\ \mathit{A}_i$, with $A_i
= \{LookCompute_i, Move_i,idle\}$ for each robot $i$. From this definition, states
are of the form $s = (s_1,\dots,s_k,c)$ where $s_i$ is the local state
of robot $i$, and $c$ the configuration. A transition of the system is labeled by a tuple $a=(a_1, \ldots, a_k)$, where
$a_i \in A_i$ for all $1 \leq i \leq k$ and $(s_1,\dots,s_k,c) \xrightarrow{a}
(s'_1,\dots,s'_k,c')$ \textit{iff} for all $i$, $s_i \xrightarrow{a_i} s_i'$ and
$c'$ is obtained from $c$ by updating the positions of all robots $i$ such that
$a_i = Move_i$. To represent the scheduling, we denote by ${\prod_{i \in
\mathit{Sched}}} \mathit{Act}_i$ the action $(a_1, \ldots, a_k)$ such that $a_i
= idle$ if $i \notin Sched$ and $a_i \in \{LookCompute_i, Move_i\}$ otherwise.

\subsection{Reachability Games}
\label{sec:jeux}
In the following we revisit the reachability games. We present here classical notions on this subject.
For more details, the interested reader can fruitfully consult the survey \cite{DBLP:conf/dagstuhl/Mazala01}.
If $A$ is a set of symbols, $A^*$ is the set of finite sequences of elements of $A$ (also called \emph{words}), and $A^\omega$
the set of infinite such sequences, with $\varepsilon$ the empty sequence. We note $A^+=A^*\setminus \{\varepsilon\}$,
and $A^\infty=A^*\cup A^\omega$. For a sequence $w\in A^\infty$, we denote its \emph{length} by $|w|$. If $w\in A^*$, $|w|$ is
equal to its number of elements. If $w\in A^\omega$, $|w|=\infty$.
For all words $w=a_1\cdots a_k\in A^*$, $w'=a'_1\cdots\in A^\infty$, we define the \emph{concatenation} of $w$ and $w'$ by the
word noted $w\cdot w'=a_1\cdots a_ka'_1\cdots$. We sometimes omit the symbol and simply write $ww'$. 
If $L\subseteq A^*$ and $L'\subseteq A^\infty$, we define $L\cdot L'=\{w\cdot w'\mid w\in L, w'\in L'\}$.

A game is composed of an \textit{arena} and \textit{winning conditions}.

\paragraph{Arena} An arena is a graph $\arena=(V,E)$ in which the set of vertices
$V=V_p\uplus V_o$ is partitioned into $V_p$, the vertices of the protagonist, and $V_o$
the vertices of the opponent. The set of edges $E\subseteq V\times V$ allows to
define the set of successors of some given vertex $v$, noted $vE=\{v'\in V\mid
(v,v')\in E\}$. 
In the following, we will only consider finite arenas.

\paragraph{Plays} To play on an arena, a token is positioned on an initial vertex. Then the token is moved by 
the players from one vertex to 
one of its successors. Each player can move the token only if it is on one of her own vertices. Formally, a play is
a path in the graph, i.e., a finite or infinite sequence of vertices $\pi = v_0v_1\dots \in V^\infty$,
where for all $0< i < |\pi|$,  $v_{i} \in v_{i-1}E$. Moreover, a play is finite only if the token
has been taken to a position without any successor (where it is impossible to continue the game): 
if $\pi$ is finite with $|\pi|=n$, then $v_{n-1}E=\emptyset$.

\paragraph{Strategies}
A strategy for the protagonist determines
to which position she will bring the token whenever it is her turn to play.
To do so, the player takes into account the history of the play, and the current
vertex.
Formally, a strategy for the protagonist is a (partial) function
$\sigma : V^*\cdot V_p \rightarrow V$ such that, for all sequence (representing the
current history) $w\in V^*$, all $v\in V_p$, $\sigma
(w\cdot v)\in vE$ (i.e. the move is possible with respect to the arena). 
A strategy $\sigma$ is \emph{memoryless} if it does not depend on the history. Formally, it means
that for all
$w, w'\in V^*$, for all $v\in V_p$, $\sigma(w\cdot v) = \sigma(w'\cdot v)$. In that case,
we may simply see the strategy as a function $\sigma:V_p\rightarrow V$.

Given a strategy $\sigma$ for the protagonist, a play $\pi=v_0v_1\cdots \in V^\infty$
 is said to be \emph{$\sigma$-consistent} if for all $0<i<|\pi|$, if $v_{i-1}\in V_p$, then $v_i=\sigma
(v_0\cdots v_{i-1})$. Given an initial
vertex $v_0$, the \emph{outcome} of a strategy $\sigma$ is the set of plays starting in $v_0$
that are $\sigma$-consistent. Formally, given an arena $A=(V,E)$, an intial vertex $v_0$ and 
a strategy $\sigma:V^*V_p\rightarrow V$, we let $\Outcome(A,v_0,\sigma)=\{v_0\pi \in V^\infty\mid v_0\pi
\textrm{ is a play and is $\sigma$-consistent}\}$.

\paragraph{Winning conditions, winning plays, winning strategies}
 We define the \emph{winning condition} for the protagonist as a subset of the plays $\Win \subseteq V^\infty$.
 Then, a play $\pi$ is \emph{winning} for the protagonist if $\pi\in \Win$. In this work, we focus on the simple
 case of reachability games: the winning condition is then expressed according to a subset of vertices $T\subseteq V$
 by $\Reach(T)=\{\pi=v_0v_1\cdots \in V^\infty\mid \exists 0\leq i < |\pi|: v_i \in T\}$. This means that the protagonist
 wins a play whenever the token is brought on a vertex belonging to the set $T$. Once it has happened, the play
 is winning, regardless of the following actions of the players.

Given an arena $\arena=(V,E)$, an initial vertex $v_0\in V$ and a
winning condition $\Win$, a
\emph{winning strategy} $\sigma$ for the protagonist is a strategy such that
any $\sigma$-consistent play is winning. In other words, a strategy $\sigma$ is winning if
$\Outcome(\arena,v_0, \sigma) \subseteq Win$. The protagonist wins the game $(\arena,v_0,\Win)$
if she has a winning strategy for $(\arena,v_0,\Win)$.
We say that $\sigma$ is winning on a subset $U \subseteq V$ if it is winning
starting from any vertex in $U$:
if $\Outcome(\arena,v_0,\sigma)\subseteq \Win$ for all $v_0\in U$.  A subset
$U\subseteq V$ of the vertices is \emph{winning} if there exists
a strategy $\sigma$ that is winning on $U$.

\paragraph{Solving a reachability game} Given an arena $\arena=(V,E)$,
a subset $T\subseteq V$, one wants to determine the set $U\subseteq V$
of winning positions for the protagonist, and a strategy $\sigma:V^*V_p\rightarrow V$
for the protagonist, that is winning on $U$ for $\Reach(T)$.

Figure~\ref{fig:synth} represents a reachability 2-player game. We recall now a well-known result on reachability games:

\begin{theorem} \label{th:memoryless}The set of winning positions for the protagonist in a reachability
game can be computed in linear time in the size of the arena. Moreover, from any position, the protagonist
has a winning strategy if and only if she has a memoryless winning strategy.
\end{theorem}

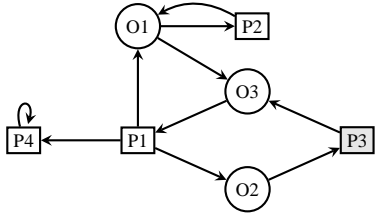
\begin{figure}[t]
\centering 
\begin{tikzpicture} [thick, node distance= 6em,>=stealth,->, scale=0.8, transform shape] 
\node[rectangle, draw, text centered](P1) {P1};
\node[rectangle, draw, text centered, left of =P1](P4){P4};
\node[circle,draw, text centered, above of = P1] (E1) []{O1};
\node[rectangle, draw, text centered, right of = E1](P2) {P2};
\node[circle,draw, text centered, above right= 1em and 4em of P1] (E3) []{O3};
\node[circle,draw, text centered, below right= 1em and 4em of P1] (E2) []{O2};
\node[rectangle, draw, text centered, above right = 1em and 4em of E2, fill=gray!20](P3) {P3};

\path (P1) edge[] (E1) 
         (P1) edge[] (P4)
	(E1) edge[] (P2) 
	(P2) edge[bend right] (E1) 
	(P1) edge[] (E2)
	(E1) edge [] (E3)
	(E2) edge[] node{} (P3)
	(P3) edge[] (E3)
	(E3) edge[] (P1)
	(P4) edge[loop above] (P4);

\end{tikzpicture} \caption{A two-player game. In this figure protagonist vertices are
represented by rectangles and antagonist vertices by circles. 
The winning condition is $\Reach(\{P3\})$. Any path in the graph is a play.
From P2 the protagonist has no winning strategy. 
From P1 a (memoryless) winning strategy is to go to O2.
Winning positions are $\{P1,P3\}$.} 
\label{fig:synth}
\end{figure}


\section{Encoding the gathering problem into a game}\label{sec:game}
As we have claimed in the introduction, the gathering problem for synchronous robots is actually a game between the robots,
that have an objective (winning condition) and evolve on a graph encoding the different configurations, and an opponent that can decide
the actual movement of a disoriented robot, i.e. a robot whose observation of the ring is symmetrical, hence is unable to distinguish its
two sides from one another. It may seem at first
that the model actually needed is the one of \emph{distributed games}, in which each robot represents a distinct player, all of them cooperating against a hostile environment. In distributed games, existence of a winning strategy
for the team of players is undecidable \cite{PetersonReif79}. However, the fact that the system is synchronous or semi-synchronous,
and that the robots are able to sense their global environment, and thus to always know the global state of the system, allows us to
stay in the framework of 2-player games, and to encode the set of robots as a single player. 
Of course, the strategy obtained will
be centralized, but we will design the game in order to obtain only strategies that can be distributed amongst anonymous, memoryless robots without chirality.
In the rest of the paper, we focus on the synchronous semantics for the system. With minor modifications, the
game can be modified to handle the semi-synchronous semantics.

\subsection{Encoding robots configurations: symmetries and equivalences}
Consider a robot system consisting of $k$ robots and $n$ nodes ($k<n$).
 The configuration of such a system is represented by the tuple $(d_1,\cdots, d_k)$, such that $\Sigma_{i=1}^k d_i=n-k$, 
and $d_i\in \{-1,0,\cdots, n-1\}$. 
Each value $d_i$ represents the number of free nodes between the $i^{\textrm{th}}$ robot and
the next robot in the clockwise direction. When the two robots occupy adjacent nodes, $d_i=0$, and when 
these two robots occupy the same node, $d_i=-1$. 
Let $\Config = \{(d_1, \cdots, d_k)\mid
\Sigma_{i=1}^k d_i=n-k\textrm{ and }d_i\in \{-1,0,\cdots,n-1\}\}$ the set of all configurations (note that $|\Config| = C^{n}_{n+k-1}$). 
In a configuration, each robot can observe the entire ring, centered in
its own position. Since the robots have no chirality,  given a configuration $C=(d_1,\cdots, d_k)$, the 
\emph{observation of robot $i$} is  $\obs_i(C)=\{(d_i, d_{i+1}, \cdots d_k, d_1, 
\cdots d_{i-1}), (d_{i-1},\cdots, d_1, d_k, \cdots d_i)\}$. Let $Obs = \{\obs_i(C)\mid C\in\Config, 1\leq i\leq k\}$ be 
the set of all possible observations.

Several types of configurations can be distinguished (see Figure~\ref{fig:views}): 
\textit{periodic}: if there are several axis of symmetry,
\textit{symmetric}: if there is only one axis of symmetry (edge-edge, node-edge, node-node),
\textit{rigid configurations}: all other configurations.

A configuration is called \emph{tower configuration} if there are several robots on the same node.  Robots constituting this tower are the ones such that at least one tuple of their 
observation begins with $-1$.

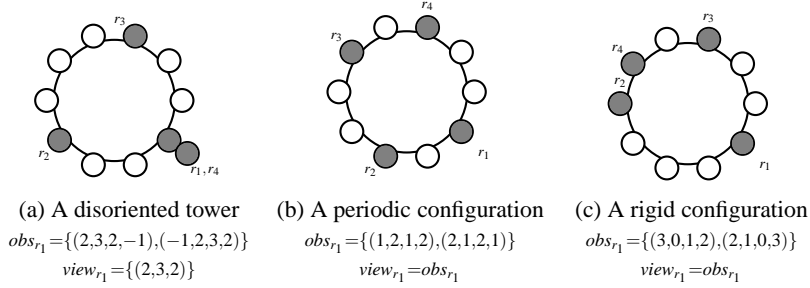
\begin{figure}[t]
\centering
\begin{subfigure}[b]{0.30\textwidth}
\captionsetup{justification=centering}
\centering
\begin{tikzpicture}[node distance =10em,>=stealth,->, scale=0.6, transform shape, thick]
  \node[circle,minimum size =8.5em, draw, thick] (C0){};

  \foreach \i in {0,...,9}
    \node[circle,minimum size =1.6em,draw,fill=white,thick]at(\i*36:1.5) (\i){};

\node[circle,minimum size =1.5em,fill=gray!100]at(2*36:1.5) (4){};
\draw node[above left of=4, node distance=1.5em] {$r_3$};
\node[circle,minimum size =1.5em,fill=gray!100]at(6*36:1.5) (6){};
\draw node[below left of=6, node distance=1.5em] {$r_2$};
\node[circle,minimum size =1.5em,fill=gray!100]at(9*36:1.5) (9){};
\node[circle,minimum size =1.6em,draw,fill=white,thick]at(9*36:2) (10){};
\node[circle,minimum size =1.5em,fill=gray!100]at(9*36:2) (10){};
\draw node[below right of=10, node distance=1.75em] {$r_1,r_4$};

\end{tikzpicture}
\caption{\small {A disoriented tower \\ ${\scriptstyle obs_{r_1}=\{(2, 3, 2, -1),(-1, 2, 3, 2)\}}$ \\ ${\scriptstyle\view_{r_1} = \{(2,3,2)\}}$}}
\label{fig:views1}
\end{subfigure}
\begin{subfigure}[b]{0.30\textwidth}
\captionsetup{justification=centering}
\centering
\begin{tikzpicture}[node distance =10em,>=stealth,->, scale=0.6, transform shape, thick]
  \node[circle,minimum size =8.5em, draw, thick] (C0){};

  \foreach \i in {0,...,9}
    \node[circle,minimum size =1.6em,draw,fill=white,thick]at(\i*36:1.5) (\i){};

\node[circle,minimum size =1.5em,fill=gray!100]at(4*36:1.5) (4){};
\node[circle,minimum size =1.5em,fill=gray!100]at(2*36:1.5) (2){};
\node[circle,minimum size =1.5em,fill=gray!100]at(7*36:1.5) (7){};
\draw node[above left of=4, node distance=1.5em] {$r_3$};
\draw node[above of=2, node distance=1.5em] {$r_4$};
\draw node[below left of=7, node distance=1.5em] {$r_2$};
\node[circle,minimum size =1.5em,fill=gray!100]at(9*36:1.5) (9){};
\draw node[below right of=9, node distance=2.25em] {$r_1$};

\end{tikzpicture}
\caption{A periodic configuration \\${\scriptstyle obs_{r_1}=\{(1,2,1,2),(2,1,2,1)\}}$\\ ${\scriptstyle \view_{r_1} = obs_{r_1}}$}
\label{fig:views2}
  \end{subfigure}
\begin{subfigure}[b]{0.30\textwidth}
\captionsetup{justification=centering}
\centering
\begin{tikzpicture}[node distance =10em,>=stealth,->, scale=0.6, transform shape, thick]
  \node[circle,minimum size =8.5em, draw, thick] (C0){};

  \foreach \i in {0,...,9}
    \node[circle,minimum size =1.6em,draw,fill=white,thick]at(\i*36:1.5) (\i){};

\node[circle,minimum size =1.5em,fill=gray!100]at(5*36:1.5) (5){};
\node[circle,minimum size =1.5em,fill=gray!100]at(4*36:1.5) (4){};
\node[circle,minimum size =1.5em,fill=gray!100]at(2*36:1.5) (2){};
\draw node[above left of=4, node distance=1.5em] {$r_4$};
\draw node[above of=2, node distance=1.5em] {$r_3$};
\draw node[above of=5, node distance=1.5em] {$r_2$};
\node[circle,minimum size =1.5em,fill=gray!100]at(9*36:1.5) (9){};
\draw node[below right of=9, node distance=2.25em] {$r_1$};

\end{tikzpicture}
\caption{A rigid configuration \\${\scriptstyle obs_{r_1}=\{(3,0,1,2),(2,1,0,3)\}}$\\ ${\scriptstyle \view_{r_1} = obs_{r_1}}$}
\label{fig:views3}
  \end{subfigure}
  \caption{Robot observations and Views}
  \label{fig:views}
\end{figure}

Since the robots take snapshots of the configuration, and their decisions are based on this information,
the states of the arena must represent the different configurations of the ring.
The robots are anonymous, hence, different rotations of a similar ring in fact represent the same configuration.
We define the rotation relation $\circlearrowright \subseteq \Config\times \Config$ as follows: for all configurations $C$, $C'\in \Config$,
$C\circlearrowright C'$ if and only if $C=(d_i,d_{i+1}, \cdots, d_{i+k-1})$ and $C'=(d_{i+1}, d_{i+2}, \cdots d_{i+k})$, where the addition symbol $+$
means sum modulo $k$.
Since the robots have no chirality, one can easily observe that, for two configurations $C$ and $C'$, if $C=(d_1,\cdots, d_k)$ and
$C'=(d_k, \cdots, d_1)$, then, for all robot $i$, $\obs_i(C)=\obs_i(C')$. We then define the mirror relation $\sim\subseteq \Config\times\Config$
by $C\sim C'$ if and only if $C=(d_1,\cdots, d_k)$ and
$C'=(d_k, \cdots, d_1)$.
From these two relations, we define an equivalence relation ${\equiv}\subseteq \Config\times\Config$ on the configurations, that identify all the configurations
on which the robots should behave the same way: we let ${\equiv} \stackrel{\mathrm{def}}{=} {(\circlearrowright\cup \sim)^*}$.

The following lemma states that our equivalence relation is correct with respect to robots behavior.
\begin{lemma}\label{lem:equivalence-classes}
For all $C\in \Config$, $\bigcup_{1\leq i\leq k} \obs_i(C) = \equivclass{C}$.
\end{lemma}

Then, an equivalence class of configurations can be seen as the set of observations for the robots in such a configuration.

We let $\equivclass{C}$ be the equivalence class of a configuration $C\in\Config$, and
we define an application $\rep:{\Config/\equiv} \rightarrow \Config$, such that 
 $\rep(\equivclass{C})\in \equivclass{C}$ for all $C\in\Config$, that associates to each 
 equivalence class a unique representative in this class, say the smallest w.r.t lexicographic order on tuples. 
For the rest of the paper, when we use the sum symbol on indexes of elements of a configuration, it means sum modulo $k$.

\subsection{Encoding the moves of the robots and transitions between configurations}\label{subsec:moves}
To define precisely the transitions between the configurations, we need the following auxiliary notations.
 We  let $\Act=\{\clockwise,\counterclockwise, \still\}$ be the different possible moves for a robot, where, as one easily 
 guesses, $\clockwise$ means that the robot moves in the clockwise direction, $\counterclockwise$ means that it moves in 
 the counter-clockwise direction, and $\still$ means that the robot does not move.
 We will use the fact that,
 for robots on a tower, a deterministic algorithm will either make them all move, or none of them. However, if they are disoriented, they
 can move in different directions.
 When a robot $i$ moves, it modifies the distances $d_i$ and $d_{i-1}$ (increasing one of these two distances by one, and
 decreasing by one the other). We can encode this by an algebraic notations, adding the configuration and one
 vector of movement for each robot: the effect on the configuration of the move $\counterclockwise$ of robot $i$ will be represented 
by the $k$-tuple $m^{i,\counterclockwise}$, the effect of the move $\clockwise$ will be represented by $m^{i,\clockwise}$ 
and if the robot does not move, it will be represented by $m^{0}$. 
These tuples are defined as follows: for a robot $1\leq i\leq k$,
$m^{i,\counterclockwise}_i=1$ and $m^{i,\counterclockwise}_{i-1}=-1$ and $m^{i,\counterclockwise}_j=0$ for all other $1\leq j\leq k$. Similarly,
$m^{i,\clockwise}_i=-1$ and $m^{i,\clockwise}_{i-1}=1$ and $m^{i,\clockwise}_j=0$ for all other $1\leq j\leq k$. The last tuple is $m^{0}_j=0$ for all 
$1\leq j\leq k$.

The idea is to add (in an element-by-element fashion) the current configuration to all the tuples representing the movements of the robots
to obtain the next configuration.
However, when the movements of two adjacent robots imply that they switch their positions
 in the ring, some absurd values (-2 or -3) may appear in the obtained configuration, if the sum is naively effected, so a careful
 treatment of these particular cases must be done. To obtain the correct configuration, one should recall
 that robots are anonymous, hence if two robots switch their positions, it has the same effect as if none of them has moved. Also,
 if in a tower, some robots want to move clockwise, and the others want to move counterclockwise, the exact robots that will move are of no importance:
 the only important thing is the number of robots that move. We will then reorganize the movements between the robots, in order
 to keep correct values in our configurations: in a tower, we will assume that the robots that will move in the counterclockwise direction
 will always be the bottom ones, and when a robot moves right and joins a tower, we will assume that it will be placed at the bottom of the tower,  and when it moves left and joins a tower, it will be placed at the top of the tower. These conventions will ensure that when adding the configuration  and the different movements, we will not obtain aberrant values.
 
 Formally, given a configuration $C=(d_1,\dots,d_k)$, we define $\PosTower(C)=\{(i,j)\mid d_j\neq -1\textrm{ and }\forall i\leq \ell< j, d_\ell=-1\}$
 that contains the positions of the towers, encoded by the position of the first and the last robot in it. We then define
$\Pos(C)=\PosTower(C) 
 \cup\{(i,i)\mid 1\leq i\leq k, \forall 1\leq \ell\leq k, (i,\ell),(\ell,i)\notin\PosTower(C)\}$, that contains the positions of the towers, and the positions of the isolated robots.
 Given a tuple of movements $(m_i)_{1\leq i\leq k}$, given $(i,j)\in\Pos(C)$, $N^\clockwise_{(i,j)}=|\{m^\clockwise_\ell\mid i\leq \ell\leq j\}|$ and 
 $N^\counterclockwise_{(i,j)}=|\{m^\counterclockwise_\ell\mid i\leq \ell\leq j\}|$. We first reorganize the movements of the robots in the towers: 
 for all $(i,j)\in\PosTower(C)$, we let $m'_\ell=m^{\ell,\counterclockwise}$ for all $i\leq \ell \leq \bigl(N_{(i,j)}^\counterclockwise +i-1\bigr)$ and 
 $m'_\ell = m^{\ell,\clockwise}$ for all $\bigl(N_{(i,j)}^\counterclockwise +i\bigr)\leq \ell \leq j$. For all $(i,i)\in\Pos(C)\setminus\PosTower(C)$, $m'_i=m_i$. Now, we iteratively
 modify the tuple $m'$. Let $(i,j)\in\Pos(C)$ be the element of $\Pos(C)$ considered at the $t^{\textrm{th}}$ iteration and let $m^t$ be the current tuple encoding the moves. 
 \begin{itemize}
\item If $d_j\neq 0$, $m^{t+1}=m^t$. 
\item Otherwise, let $r$ such that $(j+1,r)\in\Pos(C)$ (if $r=j+1$, the next robot is isolated, otherwise it is a tower). 
 \begin{itemize}
 \item If $N^{\clockwise}_{(i,j)}\geq N^{\counterclockwise}_{(j+1,r)}$, then $m^{t+1}_\ell = m^{\ell, \clockwise}$ for all 
 $j-N^\clockwise_{(i,j)}+N^\counterclockwise_{(j+1,r)}+1\leq \ell\leq j$, $m^{t+1}_\ell=m^{\ell,0}$ for all $j-N^\clockwise_{(i,j)}\leq \ell\leq j-N^\clockwise_{(i,j)}+
 N^\counterclockwise_{(j+1,r)}$ and for all $j+1 \leq \ell \leq j+N^\counterclockwise_{(j+1,r)}-1$, and $m^{t+1}_\ell=m^t_\ell$ for all other $\ell$. 
 \item If $N^{\clockwise}_{(i,j)}< N^{\counterclockwise}_{(j+1,r)}$, then the modification is symmetrical. 
 \end{itemize}
 \end{itemize}
 When all the elements of $\Pos(C)$ have been visited, we obtain a tuple $(m^f_i)_{1\leq i\leq k}$.
 
 \begin{proposition}
For all configurations $C\in\Config$, for all tuples $(m_i)_{1\leq i\leq k}$, $C+\Sum{i=1}{k} m^f_i\in\Config$,
where $(m^f_i)_{1\leq i\leq k}$ has been obtained as described above.
 \end{proposition}
 
 \begin{proof}[sketch]
 Let $C=(d_1,\dots, d_k)$. For all $1\leq i\leq k$, if $d_i=0$, then if the robot $i$ wants to move in the clockwise direction, and the robot
 $i+1$ wants to move in the counterclockwise direction, then by our construction, $m^f_i=m^{i,0}$ and $m^f_{i+1}=m^{i+1,0}$, and the resulting distance will
 stay 0. For all other decisions of the robots, the distance obtained will be positive. If $d_i=-1$, by the reorganization of the robots on a tower, it is impossible
 that robot $i$ wants to move in the clockwise direction and that the robot $i+1$ wants to move in the counterclockwise direction. Hence, the distance obtained
 is never less than -1. In all other cases, the obtained distance is necessarily positive.
 \qed\end{proof} 
 
 \begin{definition}[successor of a configuration]
 Given a configuration $C\in\Config$ and a tuple of moves for the different robots $(m_i)_{i\in\{1,\dots,k\}}\in \Act^k$, the successor configuration, noted
 $C\oplus (m_i)_{i\in\{1,\dots,k\}}$ is obtained by $C+\Sum{i=1}{k} m^f_i\in\Config$, where $(m^f_i)_{1\leq i\leq k}$ has been obtained as described above.
 \end{definition}

\subsection{The Gathering Game}
We build an arena for a reachability game, such that the protagonist has a winning strategy if and only if one can design an algorithm for the robots to gather on a single node, starting from any configuration. The possible decisions of movements taken by the robots will be noted by $\Actions=\{
\clockwise,\counterclockwise,\still,\disoriented\}$, which is the set $\Act$ of possible movements, added by a special decision \disoriented, taken by a disoriented robot that
nevertheless wants to move. We will note $\overline{\clockwise}=\counterclockwise$, $\overline{\counterclockwise}=\clockwise$, $\overline{\still}=\still$
and $\overline{\disoriented}=\disoriented$.
We consider the arena $\Agather=(V_p\uplus V_o,E)$, where the set of protagonist states is
$V_p= (\Config/\equiv)$,
the set of antagonist states is $V_o= \Config \times (\Actions^k)$, the size of the arena is thus linear in $n$ and exponential in $k$.

The edge relation $E$ will ensure a strict alternance between the two players: $E\subseteq (V_p\times V_o) \cup (V_o\times V_p)$ and
will be detailed in the rest of the subsection.
\medskip

\textbf{From $V_p$ to $V_o$}
From a protagonist position, representing an equivalence class of configurations,
the play continues on an antagonist position memorizing the different movements
decided by each robot. Such a move is possible if, in a given equivalence class
of configurations, the robots with the same observation take the same decision. 
However,
our definition of observation does not capture what happens when several
robots are stacked to form a tower: 
consider two robots on a tower, in a configuration of the form
$C=(-1,d_2,\cdots,d_k)$. Using our definition of observation, we obtain
$\obs_1(C)=\{(-1,d_2,\cdots,d_k),(d_k,\cdots,d_2,-1)\}$ and $\obs_2(C)=\{(d_2,\cdots,d_k,-1),(-1,d_k,\cdots,d_2)\}$,
hence $\obs_1(C)\neq \obs_2(C)$ whereas in reality they observe the same thing.
Thus, we will use the notion of \emph{view} for a robot, where, if a robot is part of a tower, the distance
from other robots in the tower is removed from its observation. Formally, we define the \emph{view} of the robot $i$
as follows:

\begin{definition}[view]
Let $C\in\Config$ and $1\leq i\leq k$ be a robot. Let $(d_1,\dots, d_k)\in\obs_i(C)$ be
the smallest observation of $C$, with respect to the lexicographic order. We define the \emph{view of robot $i$} by
$\view_i(C)=\{(d_i, \dots, d_j), (d_j,\dots, d_i)\}$, where $i<j$ are respectively the smallest and greatest index such that $d_i\neq -1$ 
(respectively $d_j\neq -1$).

We let $\Views = \{\view_i(C)\mid C\in\Config, 1\leq i\leq k\}$
 be the set of all possible views.
\end{definition}

Note that if robot $i$ does not belong to a tower
then $\view_i(C) = \obs_i(C)$. Also, when $|\view_i(C)|=1$,
the robot is disoriented (see Figure~\ref{fig:views}).
For $o\in Obs$ an observation, we let $p(o)\in\Views$ be the projection
from an observation to obtain a view.

A \emph{decision function} is a function that suggests a  
movement to a robot, according to its view. 

\begin{definition}[decision function] 
A \emph{decision function} is a function $f:\Views\rightarrow \Actions$ such that, for all $V\in\Views$, if $|V|=1$, then
$f(V)\in\{\still,\disoriented\}$ and if $f(V)=\disoriented$ then $|V|=1$.
\end{definition}

Given a configuration $C=(d_1,\dots,d_k)\in\Config$, we translate a decision function $f$ into a real movement of each robot. For all $1\leq i\leq k$, let $f(C,i)$ be defined as follows. If $(d_i,\cdots, d_k,d_1,\cdots d_{i-1})$ is the smallest element of $\view_i(C)=\{(d_i,\cdots, d_k,d_1,\cdots d_{i-1}),$
$(d_{i-1},\cdots, d_1,d_k,\cdots d_i)\}$ in the
lexicographic order, then $f(C,i)=f(\view_i(C))$. Otherwise, $f(C,i)=\overline{f(\view_i(C))}$.
This is so because, when applying the real movements on a real configuration, the game (that makes
the robots move) must be coherent on a common direction.

We are able to determine now the edge relation from a protagonist state to an antagonist
state: for all $v\in V_p, v'\in V_o$, 
$(v,v')\in E$ { if and only if there exists a decision function $f$ such that }$v'=\bigl (C,(a_1,\dots,a_k)\bigr )$ defined as follows:
$C=\rep(v)=(d_1,\dots,d_k)$ and, for all $1\leq i\leq k$, $a_i=f(C,i)$.
\medskip

\textbf{From $V_o$ to $V_p$}
The moves of the antagonist lead the game into the following configuration of the system resulting of the application
of the decisions of all the robots. 
If one robot decides to move, but is disoriented, then the antagonist chooses the actual move (\clockwise{} or 
\counterclockwise) the robot will make. The next configuration reached by the robots is then determined
by the actions chosen and by the decisions taken by the antagonist.

\begin{definition}\label{def:v'-comp}
For a state $v'=(C, (a_1,\dots, a_k))$, we say that a tuple $(m_i)_{i\in\{1,\dots,k\}}$ is \emph{$v'$-compatible} if, 
\begin{itemize}
\item for all $1\leq i\leq k$ such that $a_i\neq\disoriented$, $m_i=a_i$,
\item for all $1\leq i\leq k$ such that $a_i=\disoriented$, $m_i\neq\still$.
\end{itemize}
\end{definition}

A $v'$-compatible tuple is then a tuple in which the antagonist has chosen in which directions disoriented robots 
will move.

Then, we can formally define the edge relation from an antagonist state to a protagonist state: for all $v\in V_p$, $v'=(C,(a_1,\dots,a_k))\in V_o$,
$(v',v)\in E$ if and only if there exists a $v'$-compatible tuple $(m_i)_{i\in\{1,\dots,k\}}$ such that $v=\equivclass{C\oplus (m_i)_{i\in\{1,\dots,k\}}}$.

To sum up, in $\Agather$\footnote{To handle the semi synchronous semantics, the antagonist should also choose at each step the subset of robots
that will be activated.}, 
\begin{align*}
E&=\{(v,v')\in V_p\times V_o\mid \\
&\textrm{there exists a decision function $f$ such that } v'=\bigl (\rep(v),(f(C,1),\dots,f(C,k))\bigr )\}\\
&\cup \{(v',v)\in V_o\times V_p\mid v'=(C,(a_i,\dots,a_k))\\
&\textrm{ and there exists a $v'$-compatible tuple $m(_i)_{i\in\{1,\dots,k\}}$},
v=\equivclass{C\oplus (m_i)_{i\in\{1,\dots,k\}}}\}.
\end{align*}

%
%
%

We now state the result that validates the construction: solving the reachability game that we have just defined
amounts to automatically synthesizing a deterministic algorithm achieving the gathering for this system.
Let $W=\equivclass{(-1,\cdots,-1, n-1)}\in V_p$ be the equivalence class of all the configurations representing the case where all the robots
are positioned on a single node.
\begin{theorem}\label{th:correctness}
The winning region for the game $(\Agather, W)$ corresponds exactly to the set of 
configurations from which the robots can achieve the gathering.
\end{theorem}

\begin{proof}[Sketch]
An algorithm $\mathcal{F}$ can be turned into a decision function $f:\Views\rightarrow \Actions$ as follows: let $\{\mathit{view}_1,\mathit{view_2}\}\in\Views$, and assume that $\mathit{view}_1<\mathit{view}_2$ with $<$ being the lexicographic order. Let $o\in p^{-1}(\mathit{view}_1)$ be an observation compatible
with the view $\mathit{view}_1$ (we recall that $p$ is the projection of an observation for a robot in a tower to its view that removes the
elements equal to -1). Then $f(\{\mathit{view}_1,\mathit{view}_2\})=\mathcal{F}(o)$. Since the algorithm $\mathcal{F}$ takes the same decision for all the robots in a 
tower, hence for all $o\in p^{-1}(\mathit{view}_1)$, this definition indeed translates the algorithm into a decision function. The strategy that chooses this decision 
function will visit the same configurations as the algorithm on the real ring. Reciprocally, a winning strategy from a configuration class gives a decision function. To 
turn the decision functions for each configuration class into a distributed algorithm, we remark, thanks to Lemma~\ref{lem:equivalence-classes}, that one 
observation for a robot belongs to exactly one equivalence class of configurations. To determine the movement a robot takes according to its observation of the 
ring, it suffices to translate the decision function associated to the corresponding equivalence class into a movement in the ring. Then one can show that any 
sequence of configurations obtained by the algorithm corresponds to a play in the game, visiting the same configurations.
\qed\end{proof}


\section{Synthesis of 3-robots gathering protocol}
In the case of a system with three robots, there are $6$ distinct types of configuration classes:  
\begin{itemize}
	\item The 3-robots tower configuration, which is the configuration to reach: $\equivclass{(-1, -1, n-1)}$.
From this class of configuration  
the edge leads to $(C,(a_1, a_1,a_1)$ with $a_1 \in\{\still,\disoriented\}$. However, this edge is not of interest for us since the gathering property is verified.
	\item The disoriented tower is a configuration where there is an axis of symmetry passing through 
the tower and the isolated robot.
This configuration belongs to the class $\equivclass{(-1, \frac{n-1}{2}, \frac{n-1}{2})}$ and occurs only when $n$ is odd.  
In this case, all robots are disoriented and thus the outgoing edges lead to all the states $\{(-1, \frac{n-1}{2}, \frac{n-1}{2}), (a_1,a_1, a_2)\}$
with $a_1, a_2\in\{\still,\disoriented\}$.
	\item The tower configurations are the configurations of the classes $\equivclass{(-1, d_2, d_3)}$, 
with $rep(\equivclass{(-1, d_2, d_3)}) = (-1, d_2, d_3)$ and $-1<d_2< d_3 \in \N$.
The edges lead to all the states $\{(-1, d_2, d_3), (a_1,a_1,a_2\}$ with $a_1,a_2\in\{\counterclockwise,\clockwise,\still\}$.
	\item The symmetrical configurations, which is in $\equivclass{(d_1,d_1,d_2)}$ with $-1\neq d_1 \neq d_2$ and $-1 \neq d_{2}$.
 Recall that when $k$ is odd and there is an axis of symmetry, the axis goes through an occupied node.
If $d_1 < d_2$, the edges lead to $(C,(a_1,a_2, a_1)$ with $a_1\in\{\clockwise, \counterclockwise,\still\}$ and $a_2\in\{\still,\disoriented\}$, 
otherwise edges lead to $(C,(a_1,a_1, a_2)$ with $a_1\in\{\clockwise, \counterclockwise,\still\}$ and $a_2\in\{\still,\disoriented\}$.
	\item The rigid configurations are all other configurations. For a class $\mathbb{C}$ such that $\rep(\mathbb{C})=C$
does not fall into any of the above categories, 
the outgoing edges go to states $(C, (a_1,a_2,a_3))$ with $a_1,a_2,a_3\in \{\clockwise, \counterclockwise, \still\}$. 
\end{itemize}

We implemented the arena for three robots and different ring sizes, in the game-solver tool \textsc{Uppaal Tiga}~\cite{UPPAAL-TIGA}. 
We verified the impossibility of the gathering from periodic configurations.
Moreover we obtained that there is a winning strategy from all protagonist vertices except from the periodic configurations, 
and we identified in the edges relation that the edges that lead to $\{(C,(a,a,a))\}$ with $a \in \Act$ are not part of any winning strategy.

The arena without the periodic class of configuration $\{\equivclass{(d, d, d)}\}$, and the edges that lead to 
$\{(C,(a,a,a))\}$ with $a \in \Act$ from a protagonist vertex $\equivclass{C}$, is the graph such that all protagonist vertices are winning. 
In order to find the best winning strategies, weights are added on the edges. 
In order to minimize the number of robot moves, each edge is weighed by the number of robots that move.
A strategy is  a shortest path algorithm on this graph such that the protagonist vertices and opponent vertices
are handled differently.
The distance between a protagonist vertex and the configuration to reach is the minimum distance, and the distance between 
an opponent vertex and this configuration is the maximal distance between them.

We obtained all the optimal strategies, for each class of configurations 
$\equivclass{(d_1, d_2, d_3)}$, the edge relation is restricted.
From these strategies we outline the following pattern of strategy.
\begin{itemize}[parsep=0cm,itemsep=0cm,topsep=0cm] 
	\item If all robots form a tower nobody moves.
From $\equivclass{(-1, -1, n-1)}$ the edge relation leads to $((-1, -1, n-1),(\still, \still, \still))$.
	\item If $2$ robots form a tower the last robot takes the shortest path to the tower.
From $\equivclass{(-1, d_1, d_2)}$with $-1< d_1 < d_2$, the edge relation leads to 
$((-1, d_1, d_2),(\still, \still, \counterclockwise))$.
And from $\equivclass{(-1, \frac{n-1}{2}, \frac{n-1}{2})}$ the edge relation leads to 
$((-1,  \frac{n-1}{2}, \frac{n-1}{2}),(\still, \still, \disoriented))$.
	\item If the configuration is symmetrical, in$\equivclass{(d_1,d_1,d_2)}$ with $-1 < d_1 < d_2$, 
the proposed strategy depends on whether $\rep(\equivclass{(d_1,d_1,d_2)}) =  (d_1, d_1, d_2) \textit{ or }(d_2, d_1, d_1) $. 
	\begin{itemize}[parsep=0cm,itemsep=0cm,topsep=0cm] 	
		\item  If $\rep(\equivclass{(d_1,d_1,d_2)}) =  (d_1, d_1, d_2)$ then the two symmetrical robots get closer to the last robot.  
		The edge relation leads to $((d_1, d_1, d_2), (\clockwise, \still, \counterclockwise))$.
		\item  If $\rep(\equivclass{(d_1,d_1,d_2)}) =  (d_1, d_1, d_2)$ then the disoriented robot moves. 
		The edge relation leads to $((d_2, d_1, d_1), (\still, \still, \disoriented))$.
	\end{itemize}
	\item If the configuration is rigid ( in$\equivclass{(d_1,d_2,d_3)}$ with $-1< d_1< d_2<d_3$)the edge relation leads to three possibilities : 
	\begin{itemize}[parsep=0cm,itemsep=0cm,topsep=0cm] 
		\item The robot with the minimum view gets closer to its nearest neighbor. 
	In this case the edge relation leads to $((d_1, d_2, d_3), (\clockwise, \still, \still))$.
		\item The robot with the maximum view gets closer to its nearest neighbor.
	 In this case the edge relation leads to $((d_1, d_2, d_3), (\still, \still, \counterclockwise))$.
		\item The robot with the minimum view and the robot with the maximum view get
		closer to their nearest neighbor.
	In this case the edge relation leads to $((d_1, d_2, d_3),$ $(\clockwise, \still, \counterclockwise))$.
		This strategy is the two above strategies made simultaneously.
	\end{itemize}
	Thus the edge relation for rigid configuration leads to: $\{((d_1, d_2, d_3), (a_1, \still, a_2))\}$, 
	with $a_1\in \{\clockwise, \still\}$, $a_2 \in \{\still, \counterclockwise\}$ and $a_1 \neq a_2$.
\end{itemize}

From Theorem~\ref{th:correctness}, one can translate the decision functions for each configuration
into a distributed algorithm.
Among the possible strategies we present below the strategy that moves the robot with the minimum view and
the robot with the maximum view closer to their nearest neighbor in the rigid configurations.
Thus we obtain the following distributed algorithm: if the view of the robot $r$ is 
$view(r)=\{(y,-1,z),(z,-1,y)\}$ with $y<z$, $r$ robot moves in order to increment $z$ and decrement $y$. 
If $view(r)=\{(x,x,z),$ $(z,x,x)\}$ with $x<z$ then $r$ moves to increment $z$ and decrement $x$,
if $view(r)=\{(z,x,z),$ $(z,x,z)\}$ with $x<z$ then $r$ moves in any direction,
if $view(r)=\{(x,y,z),(z,y,x)\}$ with $x<y<z$ then $r$ moves to increment $z$ and decrement $x$,
if $view(r)=\{(y,x,z),(z,x,y)\}$ with $x<y<z$ then $r$ moves to increment $z$ and decrement $y$, 
and when $r$ has a different view than the above, it remains idle.
  
The above algorithm is correct by construction for various values of $n$ ($3 \leq n  \leq 15$, $n = 100$). 
The following theorem proves that  it is also correct for any ring of size $n$. Due to space limitation the proof by induction of the theorem is omitted. 
\begin{theorem}
In a ring of any size $n>3$ starting from any configuration (except periodic ones) the above 3-gathering algorithm 
eventually reaches a gathering configuration. 
\end{theorem}

\section{Conclusions and discussions}

We proposed a formal method based on reachability games that permits to automatically generate distributed algorithms for mobile autonomous robots solving a global task. The task of gathering on a ring-shaped network was used as a case study. We hereby discuss current limitations and future works.

While our construction generates algorithms for a particular number of robots $k$ and ring size $n$, the game encoding we propose enables to easily tackle the gathering problem for any given $k$ and $n$, provided as inputs, since $k$ and $n$ are parameters of the arena described in Section~\ref{sec:game}. Also, we focused on the atomic FSYNC and SSYNC models. Breaking the atomicity of Look-Compute-Move cycles (that is, considering automatic algorithm production for the ASYNC model~\cite{FPS12b}) implies that robots cannot maintain a current global view of the system (their own view may be outdated), nor be aware of the view of other robots (that may be outdated as well). Then, our two-players game encoding is not feasible anymore. A natural approach would be to use distributed games, but they are generally undecidable as previously stated. So, a completely new approach is required for the automatic generation of non-atomic mobile robot algorithms.
 
The problem of synthesis for parameterized systems is a challenging path for future research. Also, the size of the game increases quickly with the number of robots; it is expected that to-be-discovered optimizations and/or heuristics will help bringing algorithm production more practical. Finally, we believe that part of our encoding (typically, configurations and transitions between configurations) can be reused for different problems on ring-shaped networks, such as exploration with stop or perpetual exploration and easily extended to other topologies.

\bibliographystyle{abbrv}
 \bibliography{rob}

\end{document}